\documentclass[11pt,a4paper]{article}

\usepackage{eurosym}
\usepackage{amscd,amsfonts,amsmath,amssymb,amsthm,bbm,bm,latexsym,mathrsfs}
\usepackage{epsfig,graphics,graphicx,natbib,subfigure,longtable}
\usepackage[english]{babel}
\usepackage[usenames]{color,colortbl}

\usepackage{bookmark}
\usepackage{booktabs}
\usepackage{rotating}
\usepackage{wasysym}
\usepackage{lscape}
\usepackage[flushleft]{threeparttable}
\usepackage{adjustbox}

\usepackage{sgame}
\usepackage[normalem]{ulem}
\usepackage{lineno} 
\usepackage[top=1in, bottom=1in, left=0.7in, right=0.7in]{geometry}
\usepackage{setspace}
\makeatother

\theoremstyle{definition}

\newtheorem{proposition}{Proposition}

\newcommand{\cl}[1]{\mathcal{#1}}
\newcommand{\bb}[1]{\mathbb{#1}}
\newcommand{\msf}[1]{\mathsf{#1}}

\newcommand{\comillas}[1]{``\,#1\,"}

\def\eqd{\stackrel{\mbox{\scriptsize{d}}}{=}}

\begin{document}
\title{\vspace{-50pt} On a flexible construction of a negative binomial model}

\author{
\hspace{-20pt}
Fabrizio Leisen\textsuperscript{a} \hspace{15pt}
Rams\'es H. Mena\textsuperscript{b} \hspace{15pt}
Freddy Palma Mancilla\textsuperscript{b} \hspace{15pt}
Luca Rossini\textsuperscript{c}\thanks{Corresponding Author: \href{luca.rossini87@gmail.com}{luca.rossini87@gmail.com}}
 \\
 \\
        {\centering {\small
        \textsuperscript{a}University of Kent, U.K. \hspace{5pt} \textsuperscript{b}IIMAS, UNAM. Mexico \hspace{5pt} \textsuperscript{c}Vrije Universiteit Amsterdam, The Netherlands}}
     }

\date{}

\maketitle

\begin{abstract}
\noindent This work presents a construction of stationary Markov models with negative-binomial marginal distributions. A simple closed form expression for the corresponding transition probabilities is given, linking the proposal to well-known classes of birth and death processes and thus revealing interesting characterizations. The advantage of having such closed form expressions is tested on simulated and real data.
\end{abstract}


\doublespacing

\section{Introduction}
Stationary models and their properties are important in theoretical and applied problems in econometrics and in time series modeling. In particular, integer-valued Markov processes have dragged the attention of many scientists. Part of the motivation for their study is to have practical models for time-evolving counts, i.e. without the need to rely on misspecified continuous state space models such as ARMA time series models or diffusion processes. 

There are, at least, two seemingly different approaches to build integer-valued Markov models for  observations evolving in discrete or continuous time, namely the Markov chain approach \citep[e.g.][]{Karlin75} and the stochastic thinning or contraction technique. The latter is largely based on the concept of self-decomposability \citep[e.g.][]{SteutelvanHarn03} or generalizations of it \citep[e.g.][]{joe1996time,ZhuJoe10}. Indeed, for count data time series, contributions such as \citet{McKenzie86,McKenzie88}, \cite{AlOsh87}, \cite{Alzaid90}, \cite{Du91} and \cite{AlOsh92} used the thinning operator to construct AR-type models with Poisson, negative-binomial and geometric stationary distributions. A more general approach, i.e. for a wider choice of marginal distributions, is presented in the groundbreaking contribution by \cite{joe1996time}, see also \cite{jorgensen1998stationary}. For up to date accounts on the topic we refer to \cite{mckenzie2003ch} and \cite{DavisEtal2016}.

Much of the theory of stationary models, with arbitrary but given marginal distributions, has been developed for discrete-time Markov models, with the notable exception of the work set forth by \citet{barndorff2001non}, where the idea of self-decomposability is used to represent Ornstein--Uhlenbeck L\'evy-driven  stochastic differential equations.  For continuous time and discrete state-space  models, clearly the theory of Markov chains offers an excellent and  general modeling alternative. However, within such framework, {{closed form expressions for the transition probabilities are not always at hand}},  thus  complicating estimation and simulation procedures. From a statistical perspective,  the availability of computable transition probabilities is always a desirable feature. Indeed,  having stochastic equations, intensity matrices and/or expressions for Markov generators, frequently leads to overparametrized situations or numerical complications.  

Here, we present a construction of stationary Markov chains with negative-binomial marginal distributions. The construction is valid in discrete and continuous time, and has the appealing feature of having a simple and closed form expression for the corresponding transition probabilities. 
Our construction is based on the idea introduced by \cite{Pitt02} and subsequently generalized  to the continuous time case by \cite{MenaWalker09}. Their approach is based on the distributional symmetry of reversibility and is general enough to encompass other approaches based on thinning, such as the models in \cite{joe1996time}.  The idea can be recast as follows: given a desired marginal distribution,  $P_{X}$,  one
creates the dependence in the model by augmenting via 
another -arbitrary- latent variable with conditional  distribution  $P_{Y\mid X}$. With these probabilities at hand, one computes the compatible conditional distribution $P_{X\mid Y}$ via Bayes' theorem. The one-step transition probability characterizing the targeted Markov process $
P(x,A)=\mathbb{P}(X_1\in A\mid X_0=x)$ is constructed as  
\begin{eqnarray}\label{PCW02}
P(x,A)=\mathbb{E}_{Y\mid x}[P_{X\mid Y}(A)],
\end{eqnarray}
where the expectation is taken with respect to $Y\mapsto P_{Y\mid X}$ conditioned on $X=x$. This approach resembles the Gibbs sampler method used to construct reversible Markov chains and has a clear Bayesian flavor, i.e. one could speak of the prior  $P_{X}$, the likelihood model $P_{Y\mid X}$, the posterior $P_{X\mid Y}$, and the predictive \eqref{PCW02}. Generally speaking, any single-observation based predictive distribution, from an exchangeable sequence,  constitutes a well-defined Markov kernel. As such, the method is  very general,  and for this reason it has been widely used to construct time series models of the AR or ARCH-type, see e.g.  \cite{Pitt05,MenaWalker05,Contreras-Cristan09}.

There are various  ideas  to generalize the above construction to the continuous-time case, e.g. by computing the n-step ahead transition corresponding to \eqref{PCW02} and embedding it as the skeleton of a continuous time process,  or by finding the conditions such that the Chapman-Kolmogorov equations are also satisfied in continuous time. This latter approach is followed by \citet{MenaWalker09} to represent well-know continuous time  models such as various diffusion processes and Markov chain models.  Thus, in order to preserve the reversibility of the model, the idea is to allow those parameters in $P_{Y\mid X}$ -and not in $P_{X}$-, to vary on time in such a way that the  Chapman-Kolmogorov equations are satisfied.  While some appealing families of diffusion processes have been found to have representation \eqref{PCW02}, under the above extension to continuous time \citep[see, e.g.][]{Anzarut18a}, in general it is not always easy to find analytical conditions that meet Chapman-Kolmogorov equations.  Here, we unveil the  conditions to apply this construction  to negative-binomial marginal distributions, which, as we will see, represent a rich class of continuous-time Markov chains.

The model we find has a neat expression for the transition probabilities and corresponds to a well-known class of birth and death processes  \citep{Feller50,Karlin75} for which a closed form expression for  transition probabilities is not available elsewhere. This constitutes an appealing addition to the applied literature on such classes of models.

The choice of negative-binomial marginals is very appealing due to its generality, i.e. it includes the geometric distribution as a particular case and the Poisson distribution as a limiting case. Furthermore, as proven by \citet{Wolpert2011MarkovIS}, all integer-valued stationary time reversible Markov processes, whose invariant distributions are infinitely divisible, are  branching process with Poisson or negative-binomial marginal univariate distributions. Not {{surprisingly}}, discrete and continuous time models with negative-binomial marginals have been widely studied in the literature, e.g. \citet{Latour98,joe1996time,ZhuJoe03,ZhuJoe10,ZhuJoe10b}. However, except for the work by \citet{ZhuJoe10b}, where a numerical inversion technique is used to approximate the continuous time transition probability of a negative-binomial model, no closed form expression for these probabilities has been found. In particular,  we give  a closed expression for the transition probability corresponding to   \citet{ZhuJoe10b} model.  We further illustrate  the convenience of having a simple expression for transition probabilities via a simulation study and a real data analysis based on a crime reports dataset studied by \cite{Gorgi18}. The performance of our approach is compared, in terms of computational time, with the method of \cite{ZhuJoe10b} when the transition is evaluated with numerical methods.

The remaining part of this document is organized as follows: In Section~\ref{Sezione2}, we present the construction of stationary Markov chains with negative-binomial distributed marginals. Section~\ref{Sezione3} is devoted to test our approach using simulated datasets. Section~\ref{Sezione4} tackles a real data example, in particular we use the negative-binomial process to model real-time series of crime reports in the city of Blacktown in Australia. Final points and conclusions are deferred to Section~\ref{Sezione5}.

 \section{Negative Binomial Stationary Markov chain}
\label{Sezione2}

Following the construction in  \cite{MenaWalker09} described above, we  build a set of transition probabilities that characterize   a class of reversible Markov chains with negative-binomial invariant distribution. Hence, given the choice of stationary negative-binomial distribution, i.e. $X\sim\mathsf{NB}(\msf{r},\msf{q})$, with mass probabilities given by
\begin{eqnarray*}
\msf{f}_{X}(x)=\binom{x+\msf{r}-1}{x}\msf{q}^x(1-\msf{q})^\msf{r}\times\textbf{I}_{\{0,1,2,...\}}(x), \quad \mbox{with}\quad \msf{q}\in (0,1).
\end{eqnarray*}
The dependence in the model is introduced by assuming ${Y\mid X}\sim\mathsf{Bin}(X,\Theta)$. Marginalizing $X$, it follows that 
\begin{eqnarray*}
Y\sim\mathsf{NB}\Big(\msf{r}, \frac{\msf{q}\,\Theta}{1-\msf{q}(1-\Theta)}\Big),
\end{eqnarray*}
which then leads to the \comillas{posterior} distribution given by the shifted negative-binomial distribution  ${X\mid Y}\eqd Y+ \mathsf{NB}\left(Y+\msf{r},\msf{q}(1-\Theta)\right)$, with corresponding conditional distribution
\begin{eqnarray*}
\texttt{f}_{X|Y}(x|y)=\binom{x+\msf{r}-1}{x-y}\big[\msf{q}(1-\Theta)\big]^{x-y}\big[1-\msf{q}(1-\Theta)\big]^{y+\msf{r}}\,\textbf{I}_{\{y,y+1,y+2,...\}}(x).
\end{eqnarray*}
Then, one can build a reversible stochastic process $X=\{X_t\}$ whose one-step transition probabilities, $\msf{p}(x_{t-1},x_t):=\mathbb{P}[X_t=x_t\mid X_{t-1}=x_{t-1}]$, take the form 
\begin{align}\label{Transi}
\msf{p}(x_{t-1},x_t)=\sum_{y=0}^{x_{t-1}\wedge x_{t}}&\mbox{NB}(x_t-y;r+y,\mathtt{q}(1-\Theta))\mbox{Bin}(y;x_{t-1},\Theta),
\end{align}
for $x_{t-1},x_t\in \{0,1,2,...\}$.  As such, the resulting model is a well defined discrete-time, integer-valued, autoregressive-type reversible model, for which the transition \eqref{Transi} leaves  $\mathsf{NB}(\msf{r},\msf{q})$ invariant over time. A stochastic equation (SE) for this model is given by
\begin{eqnarray}
X_{n+1}=Y_n+Z_n,\label{AR-NB}
\end{eqnarray}
where $(Z_n)_{n\geq1}$ is a sequence of  random variables with common NB$(\msf{r}+Y_n,\msf{q}(1-\Theta))$ distribution, and $Y_n$ has Bin$(X_n,\Theta)$ distribution. A similar expression to the SE \eqref{AR-NB} was derived in \citet{ZhuJoe03} to characterize a reversible Markov process with negative-binomial invariant distribution. 

In order to generalize the above model to the continuous-time case we allow the {dependence parameter} $\Theta$ to depend on time, namely through a function $t\mapsto \Theta_t$, and thus find the conditions on such a function such that the corresponding transition density $\msf{p}_t$ satisfies the Chapman-Kolmogorov equation. Such task is simplified via the Laplace transform, $ \cl{L}_X(u):=\bb{E}[\msf{e}^{\msf{u}X}]$, associated to the corresponding transition probability function \eqref{Transi},  which reduces to
\begin{eqnarray}\label{CKlaplace}
\cl{L}_{X_t\mid X_0=x}(\msf{u})=\mathbb{E}\big[\msf{e}^{\msf{u}X_t}|X_0=x\big]=
\big([1-\xi_t(\msf{u})]\msf{e}^{-\msf{u}}\big)^{\msf{r}}\big[1-\Theta_t\xi_t(\msf{u})\big]^{x},\label{NB}
\end{eqnarray}
where $\xi_t(\msf{u})=(1-\msf{e}^\msf{u})/[1-\msf{q}(1-\Theta_t)\msf{e}^\msf{u}]$. Notice that, in this latter case, we have already made explicit the dependence on the time parameter. 

\begin{proposition}\label{propCK}
The transition probabilities $ \msf{p}_t$ 
satisfy the Chapman-Kolmogorov equations if and only if
$$\Theta_t=\frac{1-\msf{q}}{\msf{e}^{\msf{c}t}-\msf{q}},$$
for $t,\msf{c}>0$. In such case, $(X_t)_{t\geq0}$ turns out to be a reversible Markov process with negative-binomial stationary  distribution.
\end{proposition}

\begin{proof}
First notice that in terms of the Laplace transform, the Chapman-Kolmogorov equations corresponding to $\msf{p}_t$ are satisfied if and only if  
\begin{eqnarray}\label{CKeq}
\cl{L}_{X_{t+s}\mid X_0=x}(\msf{u})=\bb{E}\left[\cl{L}_{X_{t+s}\mid X_s}(\msf{u})\mid X_0=x\right], 
\end{eqnarray}
where the time-effect in the law of $\{X_{t+s}\mid X_s\}$ enters non-homogenously, i.e. through a function $\Theta_t$. Hence, using \eqref{CKlaplace}, the right-hand side of \eqref{CKeq} takes the form
\small
\begin{eqnarray}
\bigg[\frac{[1-\xi_t(\msf{u})][1-\xi_s(\log[1-\Theta_t\xi_t(\msf{u})])]}{\msf{e}^\msf{u}\big[1-\Theta_t\xi_t(\msf{u})\big]}\bigg]^{\msf{r}}
\Big[1-\Theta_s\xi_s\big(\log[1-\Theta_t\xi_t(\msf{u})]\big)\Big]^{x}.\label{CK1}
\end{eqnarray}
\normalsize
On the other hand, the left-hand side of \eqref{CKeq} is given by
\begin{eqnarray}
\big([1-\xi_{t+s}(\msf{u})]\msf{e}^{-\msf{u}}\big)^{\msf{r}}\big[1-\Theta_{t+s}\xi_{t+s}(\msf{u})\big]^{x}.\label{CK2}
\end{eqnarray}
 Therefore, the Chapman-Kolmogorov equations are fulfilled as long as the following equalities hold,
\begin{eqnarray}
\frac{1-\msf{q}(1-\Theta_s)(1-\Theta_t)}{[1-\msf{q}(1-\Theta_t)][1-\msf{q}(1-\Theta_s)]}=\frac{1}{1-\msf{q}(1-\Theta_{t+s})}\label{1},
\end{eqnarray}
and
\begin{eqnarray}
\frac{1-\Theta_s\Theta_t-\msf{q}(1-\Theta_t)(1-\Theta_s)}{[1-\msf{q}(1-\Theta_t)][1-\msf{q}(1-\Theta_s)]}=\frac{1-\Theta_{t+s}}{1-\msf{q}(1-\Theta_{t+s})}.\label{2}
\end{eqnarray}
Those equations are satisfied whenever
\begin{eqnarray*}
\psi_t\psi_s:=\frac{\Theta_s\Theta_t}{[1-\msf{q}(1-\Theta_t)][1-\msf{q}(1-\Theta_s)]}=\frac{\Theta_{t+s}}{1-\msf{q}(1-\Theta_{t+s})}=:\psi_{t+s},
\end{eqnarray*}
\normalsize
whose solution is given by $\psi_t=\msf{e}^{-\msf{c}t}$ for $\msf{c}\in\mathbb{R}$. Thus, $\Theta_t=(1-\msf{q})/(\msf{e}^{\msf{c}t}-\msf{q})$ for $\msf{c}>0$, given that   $\Theta_t:\mathbb{R}_+\rightarrow[0,1]$.
Therefore, for this  form of $\Theta_t$, the transition probabilities $\msf{p}_t$ satisfy the Chapman-Kolmogorov equations.
\end{proof}

\noindent Using Proposition~\ref{propCK} a closed form expression for the transition probabilities simplify as
\begin{eqnarray}\label{cttran}
\msf{p}_t(x,x_t)
 =\sum_{y=0}^{x\wedge x_t}\binom{x_t+r-1}{x_t-y}\binom{x}{y}\Big[\msf{q}(1-\Theta_t)\Big]^{x_t-y} \Big[1-\msf{q}(1-\Theta_t)\Big]^{r+y} \Theta_t^{y} \Big(1-\Theta_t\Big)^{x-y}.
\end{eqnarray}
Given the discrete state-space nature of the model we have just constructed, it is natural to think that it should coincide with a continuous-time Markov chain. Indeed, L'H\^{o}pital's rule allows us to compute the infinitesimal rates associated to $\msf{p}_t$ as follows
$$\lim_{t\downarrow0}\frac{\msf{p}_t(i,j)}{t}=\lim_{t\downarrow0}\frac{\partial}{\partial t}\msf{p}_t(i,j).$$
Now, since $\lim_{t\downarrow0}\Theta_t=1$ then the limit of most of the summands in $(\partial/\partial t)\msf{p}_t(i,j)$, when $t\downarrow0$, are zero. Using the equality  $\lim_{t\downarrow0}\frac{\partial\Theta_t}{\partial t}=-\frac{\msf{c}}{1-\msf{q}}$ we obtain
\begin{eqnarray*}
\lim_{t\downarrow0}{\frac{\msf{p}_t(i,j)}{t}}=\begin{cases}
\frac{\msf{c}i}{1-\msf{q}}&\mbox{if }j=i-1,\\
\frac{\msf{c}\msf{q}(i+\msf{r})}{1-\msf{q}}&\mbox{if }j=i+1,\\
-\frac{\msf{c}i}{1-\msf{q}}-\frac{\msf{c}\msf{q}(i+\msf{r})}{1-\msf{q}}&\mbox{if }j=i,\\
0&\mbox{for }j\neq i-1,i,i+1.
\end{cases}
\end{eqnarray*}
\color{black} 
Therefore, the process $(X_t)_{t\geq 0}$ reduces to a simple birth, death and immigration process with  death \color{black} rate $\mu=\msf{c}/(1-\msf{q})$,  birth \color{black} rate $\lambda=\msf{c}\msf{q}/(1-\msf{q})$ and immigration rate $\nu=\msf{c}\msf{q} \msf{r}/(1-\msf{q})$, or equivalently, $\msf{q}=\lambda/\mu$, $\msf{r}=\nu/\lambda$ and $\msf{c}=\mu-\lambda$ (\emph{cf.} Kelly, 2011). It is worth noticing that, since $\msf{c}>0$, we have that $\mu>\lambda$, which is consistent with the model. 
Alternatively, we can characterize the above process as the birth and death process with rates $\lambda_n=\lambda n+\nu$ and $\mu_n=\mu n$, also known as linear growth process (\emph{cf.} Karlin, 1975). To the best of our knowledge expression \eqref{cttran} is new in the literature. In particular, if $\msf{r}=1$, $(X_t)_{t\geq 0}$ has a geometric stationary distribution. In such case, the rates are given by $\mu_n=\mu n$ and $\lambda_n=\lambda n+\lambda$, implying that $\lambda=\nu$, which is a variation of the simple queue, $M/M/1$, model characterized as the birth and death process with rates $\mu_n=\mu$ and $\lambda_n=\lambda$. 

Furthermore, generalizing the stochastic equation \eqref{AR-NB}, we deduce that the birth, death and immigration process satisfies the following generalized branching operation
\begin{eqnarray}
X_t=\sum_{j=0}^{X_0}{I_j(\Theta_t)}+Z_t,\label{SE}
\end{eqnarray}
where $I_0=0$, $I_j(\alpha)$s are i.i.d. random variables with common NB$(1,\alpha)$ distribution and $Z_t$ follows a  NB$(\mathtt{r}+Y_t,\mathtt{q}(1-\Theta_t))$, with $Y_t=\sum_{j=0}^{X_0}{I_j(\Theta_t)}$. That is to say, the model enjoys the branching property \citep{Athreya2006}.   Using a different approach, 
\citet{ZhuJoe03} derived a different stochastic equation representation for which no closed expression for the transition function is available. In \cite{ZhuJoe10b}, a numerical inversion technique of the corresponding characteristic function is used for an approximation of transition probabilities. 
Stochastic equation \eqref{SE}, together with  closed form expressions for  transition  probabilities, can be used for the study, estimation and simulation purposes.

Indeed, equation \eqref{cttran} completes the availability of transition probabilities for the class of continuous-time reversible Markov processes with non-negative integer values, that can  be defined via a thinning operation \citep[see][]{Wolpert2011MarkovIS}.
As a twofold observation,  such transitions turn out to be finite sums of positive terms, which can simplify some  simulation and estimation procedures. 

\section{Simulation Experiment}
\label{Sezione3}
In order to  evaluate the performance of our construction, we test the model with different simulated examples. {We consider two different datasets simulated from \eqref{SE},  for the first dataset we assume equally spaced data, i.e.  $\tau_n = t = 1$, while the second dataset features data generated at exponential times with intensity parameter $\lambda = 0.5$.  We also consider two different simulation scenarios: the first one consists in a single dataset of $1000$ observations; the second one consists in  $100$ different datasets of $1000$ observations. }In both cases, we run single experiments taking in consideration the first $250$ observations, the first $500$ and the full dataset. We keep fixed different values of the parameters of interest, in particular, we are interested in the estimation of the parameters of the negative-binomial distribution, $(\msf{q},\msf{r})$, namely the probability of success in each experiment, $\msf{q}\in (0,1)$, and the number of failures until the experiment is stopped ($\msf{r}\ge 0$), respectively.  In addition to these two parameters, we need to estimate the dependency parameter $\msf{c}$.

A Maximum Likelihood Estimation (MLE) approach is adopted in both scenarios. We set the parameter $\msf{c}$ equal to $0.5$ and $1$. Regarding the negative binomial distribution, we set the probability of success parameter in each experiment to $\msf{q} = 0.3, 0.5 $ and $0.7$, and the number of failures parameter, $\msf{r}$,  to $2$ and $5$.
In Tables~\ref{tab1} and \ref{tab2}, we report the results for different values of $\msf{r}$, $\msf{q}$ and $\msf{c}$ for a single chain and over $100$ different datasets respectively, {when the data are assumed to be equally spaced}. In particular, we report the means over the $100$ experiments and the standard deviations (in brackets). Clearly considering the full dataset improves the results, opposed to the case with only a portion of it (i.e. $250$ or $500$ observations).  {{The same analysis has been conducted for data generated at exponential times with intensity parameter $\lambda= 0.5$}}. In Table~\ref{Simu_Iter_Exp}, we report the means and the standard deviations over $100$ different datasets simulated at exponential times.
Therefore, the evidence suggests the estimation method works correctly for the negative-binomial Markov chain and we can proceed by applying it to real data experiment, as we do in the following section.

\begin{table}[t]
\centering
\begin{small}
\begin{tabular}{l|ccc|ccc|ccc}
\hline 
  & $\msf{r}$ & $\msf{q}$ & $\msf{c}$ & $\msf{r}$ & $\msf{q}$ & $\msf{c}$  & $\msf{r}$ & $\msf{q}$ & $\msf{c}$\\
  \multicolumn{1}{c}{\textit{True Value}} & 2 & 0.30 & 0.50 & 2 & 0.50 & 0.50  & 2 & 0.70 & 0.50 \\
\hline 
$T = 250$ & 2.3804 & 0.2864 & 0.4721 &  2.6790 & 0.4468 & 0.5667 & 2.5355 & 0.6400 & 0.6073 \\
$T = 500$ &  2.8200 & 0.2445 & 0.5548 &  2.9365 & 0.4337 & 0.5282 & 2.1239 & 0.6949 & 0.5553 \\
$T = 1000$ &  2.0177 & 0.3103 & 0.4854 & 2.0756 & 0.4980 & 0.4939 & 2.0436 & 0.6937 & 0.5072 \\
\hline 
  \multicolumn{1}{c}{\textit{True Value}} & 2 & 0.30 & 1 & 2 & 0.50 & 1 & 2 & 0.70 & 1 \\
  \hline
$T = 250$ & 2.1469 & 0.2732 & 1.1226 & 1.6836 & 0.5787 & 1.0903  & 1.9114 & 0.7086 & 0.7839 \\
$T = 500$ &  2.0623 & 0.2894 & 0.9673 & 2.0326 & 0.5291 & 1.0183 & 2.0758 & 0.6979 & 0.8350 \\
$T = 1000$ &  2.0100 & 0.2819 & 1.0253 &  2.0574 & 0.5114 & 1.0163 & 2.0578 & 0.6832 & 0.9758 \\
\hline 
  \multicolumn{1}{c}{\textit{True Value}} & 5 & 0.30 & 0.50 &  5 & 0.50 & 0.50 &  5 & 0.70 & 0.50 \\
  \hline 
$T = 250$ &  6.9482 & 0.2568 & 0.4556 & 5.4891 & 0.4695 & 0.5093 & 4.4857 & 0.7421 & 0.5067 \\
$T = 500$ &  6.6369 & 0.2501 & 0.4672 & 5.9744 & 0.4463 & 0.5571 & 4.3819 & 0.7364 & 0.4583 \\
$T = 1000$ & 5.0930 & 0.2985 & 0.4930 & 5.0801 & 0.4773 & 0.5654 & 5.0734 & 0.6938 & 0.5212 \\
\hline
  \multicolumn{1}{c}{\textit{True Value}} & 5 & 0.30 & 1 & 5 & 0.50 & 1 & 5 & 0.70 & 1 \\
  \hline 
$T = 250$ & 5.4157 & 0.3010 & 0.9969 & 6.2832 & 0.4275 & 1.5045 & 4.4588 & 0.7227 & 0.8783 \\
$T = 500$ & 5.7240 & 0.2872 & 1.0208 & 4.8124 & 0.4991 & 1.1054 & 4.9064 & 0.7015 & 0.8917 \\
$T = 1000$ & 5.0294 & 0.3101 & 0.9307 & 5.0545 & 0.4897 & 0.9837 & 4.9891 & 0.7002 & 0.9844 \\
  \hline 
\end{tabular}
\caption{Maximum Likelihood Estimation for different sample sizes $T$ and for different parameter values $(\msf{r},\msf{q},\msf{c})$ for one experiment and for equally spaced data.\label{tab1}}
\end{small}
\label{Simu_1Chain}
\end{table}

\begin{table}[h!]
\centering
\begin{scriptsize}
\begin{tabular}{l|ccc|ccc|ccc}
\hline 
  & $\msf{r}$ & $\msf{q}$ & $\msf{c}$ & $\msf{r}$ & $\msf{q}$ & $\msf{c}$  & $\msf{r}$ & $\msf{q}$ & $\msf{c}$\\[0.04cm] 
  \multicolumn{1}{c}{\textit{True Value}} & 2 & 0.30 & 0.50 & 2 & 0.50 & 0.50 & 2 & 0.70 & 0.50 \\
\hline
$T = 250$ & 5.1924 & 0.2831 & 0.5298 & 2.1449 & 0.4924 & 0.5257 & 2.0586 & 0.6953 & 0.5112 \\ 
& (18.8061) & (0.1119) & (0.1158) & (0.6582) & (0.0763) & (0.1046) &  (0.3654) & (0.0427) & (0.0897) \\
$T = 500$ & 4.1335 & 0.2851 & 0.5113 &  2.0658 & 0.4944 & 0.6959 &  2.0397 & 0.6960 & 0.5075 \\ 
& (16.5762) & (0.0908) & (0.0789) &  (0.4179) & (0.0583) & (1.8185) & (0.2776) & (0.0309) & (0.0633) \\
$T = 1000$ & 2.2101 & 0.2909 & 0.5064 &  2.0330 & 0.4963 & 0.5085 &  2.0543 & 0.6949 & 0.5077 \\ 
& (0.7434) & (0.0583) & (0.0527) & (0.2559) & (0.0369) & (0.0529) & (0.1811) & (0.0208) & (0.0442) \\
\hline 
  \multicolumn{1}{c}{\textit{True Value}} & 2 & 0.30 & 1 & 2 & 0.50 & 1 & 2 & 0.70 & 1 \\
  \hline 
  $T = 250$ &  3.9394 & 0.2951 & 1.0471 & 2.0887 & 0.4954 & 1.0664 & 2.0367 & 0.6959 & 1.0327 \\ 
& (16.5737) & (0.0838) & (0.1967) & (0.5116) & (0.0586) & (0.1853) & (0.3389) & (0.0406) & (0.1966) \\
  $T = 500$ & 2.1304 & 0.2976 & 1.0227 & 2.0953 & 0.4940 & 1.0201 & 2.0202 & 0.6960 & 1.0197 \\ 
& (0.6029) & (0.0576) & (0.1520) & (0.3552) & (0.0457) & (0.1362) & (0.2007) & (0.0263) & (0.1264) \\ 
  $T = 1000$ & 2.0593 & 0.2991 & 1.0018 & 2.0367 & 0.4993 & 1.0018 & 2.0385 & 0.6954 & 1.0149 \\ 
& (0.3889) & (0.0438) & (0.0915) & (0.2715) & (0.0352) & (0.0892) & (0.1592) & (0.0189) & (0.0915) \\
\hline 
  \multicolumn{1}{c}{\textit{True Value}} & 5 & 0.30 & 0.50  & 5 & 0.50 & 0.50 & 5 & 0.70 & 0.50 \\
  \hline
  $T = 250$  &  10.0632 & 0.2617 & 0.5287 & 5.5172 & 0.4862 & 0.5165 & 5.2965 & 0.6897 & 0.5176 \\
& (22.6397) & (0.0974) & (0.0948) & (1.6552) & (0.0718) & (0.1014) & (1.1201) & (0.0451) & (0.0901) \\
  $T = 500$  & 5.9747 & 0.2823 & 0.5134 & 5.2108 & 0.4947 & 0.5069 & 5.1213 & 0.6968 & 0.5091 \\ 
& (2.6223) & (0.0692) & (0.0634) & (0.9452) & (0.0477) & (0.0638) & (0.7501) & (0.0307) & (0.0640) \\ 
  $T = 1000$  &5.5866 & 0.2838 & 0.5152 & 5.0699 & 0.4992 & 0.4986 & 5.0455 & 0.6992 & 0.5047 \\ 
& (1.3622) & (0.0456) & (0.0486) & (0.5839) & (0.0301) & (0.0403) & (0.4948) & (0.0210) & (0.0420) \\
\hline 
  \multicolumn{1}{c}{\textit{True Value}} & 5 & 0.30 & 1 & 5 & 0.50 & 1 & 5 & 0.70 & 1 \\
  \hline 
  $T = 250$ & 6.1472 & 0.2842 & 1.0795 & 5.3822 & 0.4847 & 1.0574 & 5.1757 & 0.6949 & 1.0442 \\
& (3.4029) & (0.0736) & (0.1740) & (1.0049) & (0.0486) & (0.1976) & (0.6896) & (0.0284) & (0.1577) \\
  $T = 500$ &  5.6657 & 0.2872 & 1.0382 & 5.2481 & 0.4894 & 1.0327 & 5.0804 & 0.6976 & 1.0281 \\
& (1.9795) & (0.0572) & (0.1207) & (0.7683) & (0.0373) & (0.1202) & (0.5299) & (0.0221) & (0.1020) \\
  $T = 1000$ &  5.4636 & 0.2864 & 1.0259 & 5.1530 & 0.4933 & 1.0156 & 5.0521 & 0.6985 & 1.0214 \\ 
   &(1.1945) & (0.0412) & (0.0848) & (0.5224) & (0.0286) & (0.0965) &(0.3779) & (0.0157) & (0.0790) \\
  \hline 
\end{tabular}
\caption{Maximum Likelihood Estimation for different sample sizes, $T$, and for different parameter values $(\msf{r},\msf{q},\msf{c})$, over $100$ different simulations for equally spaced data.\label{tab2}}
\end{scriptsize}
\label{Simu_Iter}
\end{table}

\begin{table}[h!]
\centering
\begin{scriptsize}
\begin{tabular}{l|ccc|ccc|ccc}
\hline 
  & $\msf{r}$ & $\msf{q}$ & $\msf{c}$ & $\msf{r}$ & $\msf{q}$ & $\msf{c}$  & $\msf{r}$ & $\msf{q}$ & $\msf{c}$\\
  \multicolumn{1}{c}{\textit{True Value}} & 2 & 0.30 & 0.50 & 2 & 0.50 & 0.50 & 2 & 0.70 & 0.50 \\
\hline 
  $T = 250$ & 4.6510 & 0.2704 & 0.6571 & 2.1959 & 0.4853 & 0.6936 & 2.0938 & 0.6933 & 0.6990 \\ 
& (16.6464) & (0.1039) & (0.1036) & (0.6127) & (0.0707) & (0.1325) & (0.4131) & (0.0466) & (0.1225) \\
  $T = 500$ & 2.3150 & 0.2853 & 0.6468 & 2.0589 & 0.4977 & 1.8029 & 2.0238 & 0.7005 & 0.6789 \\ 
& (0.7978) & (0.0660) & (0.0820) & (0.4692) & (0.0521) & (4.5160) & (0.2782) & (0.0290) & (0.0898) \\ 
  $T = 1000$ & 2.1454 & 0.2939 & 0.6396 & 2.0396 & 0.4960 & 0.6646 & 2.0035 & 0.7017 & 0.6695 \\ 
& (0.5590) & (0.0463) & (0.0553) & (0.2690) & (0.0337) & (0.0581) & (0.2035) & (0.0213) & (0.0620) \\
\hline 
  \multicolumn{1}{c}{\textit{True Value}} & 2 & 0.30 & 1 & 2 & 0.50 & 1 & 2 & 0.70 & 1 \\
  \hline
  $T = 250$ & 5.2547 & 0.2882 & 1.0930 & 2.1036 & 0.4928 & 1.1165 & 2.1026 & 0.6900 & 1.1219 \\ 
& (20.5733) & (0.0962) & (0.2365) & (0.5095) & (0.0609) & (0.2205) & (0.3163) & (0.0374) & (0.2281) \\
  $T = 500$ & 2.3421 & 0.2882 & 1.0730 & 2.0841 & 0.4926 & 1.0943 & 2.0280 & 0.6976 & 1.0770 \\ 
& (1.1416) & (0.0713) & (0.1621) & (0.3529) & (0.0436) & (0.1541) & (0.2246) & (0.0274) & (0.1502) \\ 
  $T = 1000$ & 2.1014 & 0.2976 & 1.0640 & 2.0347 & 0.4958 & 1.0844 & 1.9933 & 0.7004 & 1.0767 \\ 
& (0.5394) & (0.0488) & (0.1107) & (0.2313) & (0.0279) & (0.1052) & (0.1611) & (0.0190) & (0.0975) \\
\hline 
  \multicolumn{1}{c}{\textit{True Value}} & 5 & 0.30 & 0.50  & 5 & 0.50 & 0.50 & 5 & 0.70 & 0.50 \\
  \hline 
  $T = 250$  &  8.4702 & 0.2829 & 0.6920 & 5.5603 & 0.4872 & 0.6865 & 5.1636 & 0.6961 & 0.7036 \\ 
& (16.8198) & (0.0981) & (0.1304) & (1.8687) & (0.0760) & (0.1423) & (1.0704) & (0.0396) & (0.1317) \\
$T = 500$  & 5.5907 & 0.2916 & 0.6687 & 5.3010 & 0.4905 & 0.6853 & 5.1234 & 0.6963 & 0.6959 \\ 
& (1.9261) & (0.0635) & (0.0859) & (1.0392) & (0.0504) & (0.1035) & (0.6104) & (0.0256) & (0.0799) \\ 
$T = 1000$  & 5.4972 & 0.2898 & 0.6628 & 5.1664 & 0.4960 & 0.6774 & 5.0589 & 0.6985 & 0.6933 \\ 
& (1.3966) & (0.0496) & (0.0543) & (0.8282) & (0.0395) & (0.0684) & (0.4402) & (0.0186) & (0.0651) \\
\hline 
  \multicolumn{1}{c}{\textit{True Value}} & 5 & 0.30 & 1 & 5 & 0.50 & 1 & 5 & 0.70 & 1 \\
  \hline 
  $T = 250$ & 7.3376 & 0.2749 & 1.1729 & 5.2690 & 0.4952 & 1.1158 & 5.1834 & 0.6933 & 1.1725 \\
& (6.9890) & (0.0905) & (0.2769) & (1.3803) & (0.0632) & (0.2066) & (0.8076) & (0.0331) & (0.2278) \\
$T = 500$  & 5.7321 & 0.2837 & 1.1232 & 5.0035 & 0.5035 & 1.0971 & 5.0477 & 0.6984 & 1.1411 \\
& (1.6732) & (0.0565) & (0.1638) & (0.8465) & (0.0412) & (0.1548) & (0.5084) & (0.0224) & (0.1550) \\
$T = 1000$  & 5.3743 & 0.2924 & 2.3472 & 5.0271 & 0.5011 & 1.0922 & 5.0267 & 0.6994 & 1.1170 \\
& (1.1959) & (0.0439) & (5.0475) & (0.5372) & (0.0269) & (0.0993) & (0.3634) & (0.0158) & (0.1032) \\
  \hline 
\end{tabular}
\caption{Maximum Likelihood Estimation for different sample sizes, $T$, and for different parameter values, $(\msf{r},\msf{q},\msf{c})$, over $100$ different simulations for exponential times data.\label{tab3}}
\end{scriptsize}
\label{Simu_Iter_Exp}
\end{table}

\section{Application to Crime Data}
\label{Sezione4}

This section is devoted to the study of the monthly number of offensive conduct reported in the city of Blacktown, Australia, from January 1995 to December 2014. Following \cite{Gorgi18}, we employ our time series approach to the New South Wales (NSW) dataset of police reports provided by the NSW Bureau of Crime Statistics and Research and currently available at \href{http://www.bocsar.nsw.gov.au/}{http://www.bocsar.nsw.gov.au/}. 

Figure \ref{Real_Plot} shows the plot of the time series and the presence of two peaks, related to high level of criminal activities in 2002 and 2010. Looking at the data, the sample mean and variance are $9.2625$ and $24.253$, respectively. As suggested by \cite{Gorgi18}, this indicates over-dispersion, and thus a negative binomial distribution for the error term is a  suitable assumption for the data.

\begin{figure}[h!]
\centering
\begin{tabular}{cc}
{\includegraphics[width = 5.5cm]{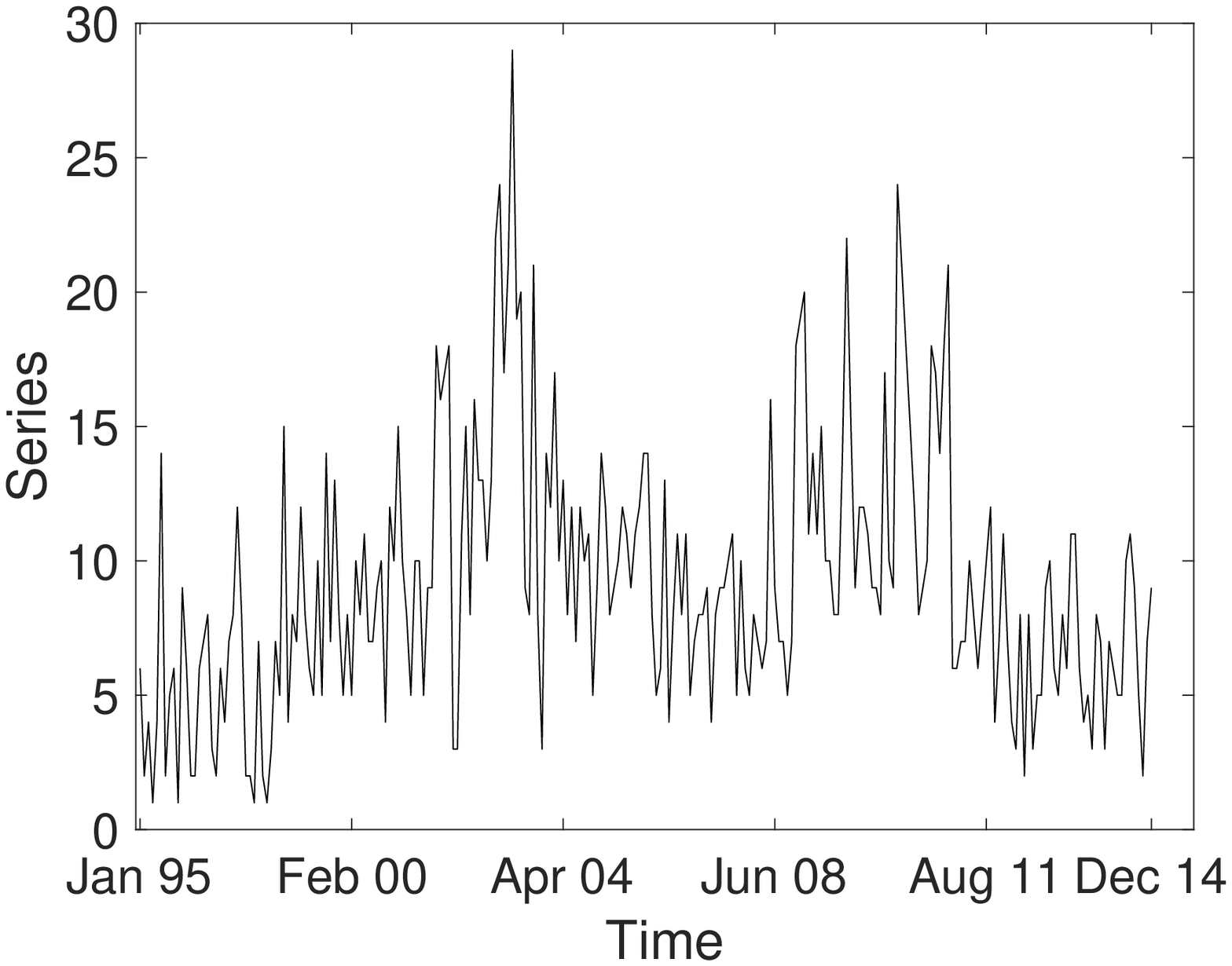}} & 
{\includegraphics[width = 5.5cm]{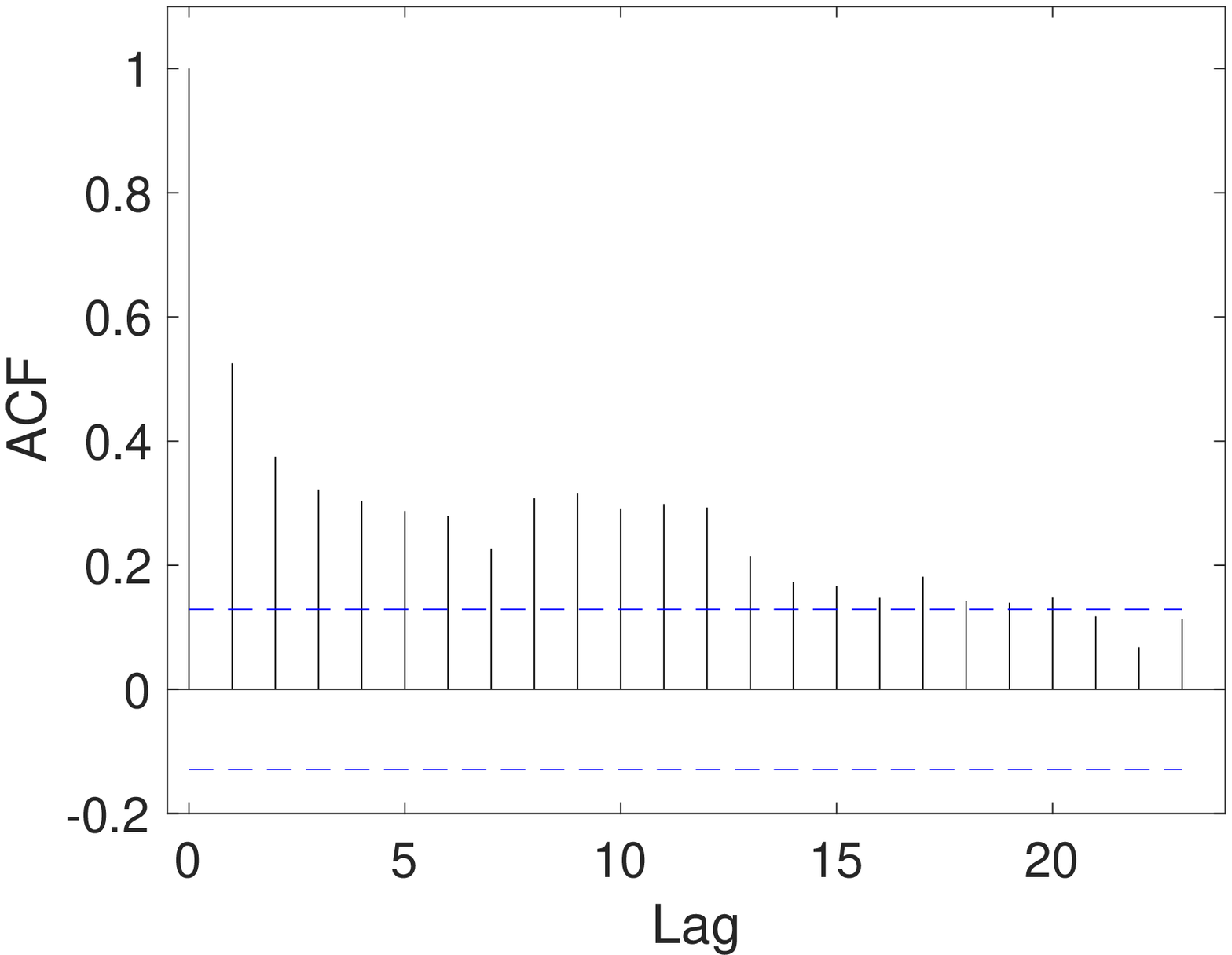}} \\ 
\end{tabular}
\caption{Monthly number of offensive conduct reports in Blacktown from January 1995 to December 2014 (left panel) and sample autocorrelation functions of the series (right panel).}
\label{Real_Plot}
\end{figure}

In Table \ref{Real_Comp}, we compare the estimation values of $(\msf{r},\msf{q},\msf{c})$ and the computational times of our exact transition probability ({\sf{Exact}}), Joe's transition probability by computing numerically the integral in Matlab ({\sf{Numerical Inversion 1}}) and Joe's transition probability by using an arbitrary precision integral ({\sf{Numerical Inversion 2}}).  The results confirm that the three methods lead to the same results, but our transition probability is much faster than the other two methods, moving from 17 seconds to 46 minutes (for the numerical integration) and 5 hours (for the arbitrary precision integral). Due to the high time consuming algorithm, in the forecasting approach, we decide to run only our method.

\begin{table}[h!]
\centering
\begin{tabular}{ccc}
\hline
Transition probability & $(\msf{r},\msf{q},\msf{c})$ & Time (in s) \\
\hline 
{\sf{Exact}} & (6.0865, 0.6031, 0.6848) & 17.65 sec \\
{\sf{Numerical Inversion 1}} & (6.0865, 0.6486, 0.7237) & 2575.60 sec \\ 
{\sf{Numerical Inversion 2}} & (6.5110, 0.4128, 0.9002) & 17085.02 sec \\ 
\hline
\end{tabular}
\caption{Comparison of the estimation values and computational timing (in seconds) between our transition probability ({\sf{Exact}}) and the Joe's transition probability ({\sf{Numerical Inversion 1}} and {\sf{Numerical Inversion 2}}). \label{tab4}}
\label{Real_Comp}
\end{table}

\noindent We perform an  out-of-sample experiment to compare our results with the forecast performances in \cite{Gorgi18}. We further divide the time series into two subsamples: the first $140$ observations are the in-sample analysis and the last $100$ observations are the out-of-sample analysis or forecasting evaluation sample. In particular, the in-sample analysis is expanded recursively.  The accuracy of the forecasting procedure is measured in terms of both point and density forecasting capabilities. Hence, we evaluate the point forecast accuracy via the mean square error (MSE) given by 
\begin{equation}
\text{MSE} = \frac{1}{100} \sum_{i=1}^{100} \left(\hat{y}_{T+i} - y_{T+i}\right)^2. \label{MSE}
\end{equation}
On the other hand, we measure the density forecasting accuracy by means of the log predictive score criterion given by  
\begin{equation}
\text{PL} = \frac{1}{100} \sum_{i=1}^{100} \log{\hat{p}_{T+i|T+i-1}(y_{T+i}}). \label{PL}
\end{equation}

\noindent Table \ref{Tab_real} shows the results in terms of point and density forecasting at different horizons $h$. We note that our model performs well, when compared to the different models proposed in \cite{Gorgi18}, in  particular when compared to the GAS-NBINAR model. Overall, we can conclude that the negative-binomial continuous-time Markov process performs well for estimation and forecasting purposes.

\begin{table}[h]
\centering
\begin{tabular}{ccccc}
\hline \\[-0.4cm]
& h=1 & h=2 & h=3 & h=4 \\[0.04cm]
\hline \\[-0.4cm]
\text{Mean square error} & 15.461 & 17.627 & 19.794 & 21.138 \\[0.04cm]
\text{Log score criterion} & -2.731 & -2.797 & -2.857 & -2.873 \\
\hline
\end{tabular}
\caption{Forecast mean square error and log score criterion computed using the last $100$ observations for different forecast horizons h.}
\label{Tab_real}
\end{table}

\section{Conclusions}
\label{Sezione5}

We present a {{closed}} form expression for the transition density of a negative-binomial continuous time Markov model. This model can be used in both, discrete and continuous time, and it is of great interest when modeling integer-valued time series. Our construction results in simpler simulation and estimation procedures.  Furthermore, it links nicely with models coming from the continuous time Markov chains literature, thus opening a gateway for its use also in that area. 

\section*{Acknowledgements}
{{We thank the Editor, Associate Editor and an anonymous reviewer for the useful comments which significantly improved the presentation and quality of the paper}}. Fabrizio Leisen was supported by the European Community's Seventh Framework Programme [FP7/2007-2013] under
grant agreement no: 630677. Ramses H. Mena gratefully acknowledges the support of CONTEX project 2018-9B.  
Luca Rossini has received funding from the European Union’s Horizon 2020 research and innovation programme under the Marie Sklodowska-Curie grant agreement No 796902.

\bibliographystyle{apalike}
\bibliography{NbBiblio}

\end{document}